\documentclass{amsart}
\usepackage{amsmath, amsthm, amssymb, amsfonts, graphicx}
\usepackage{mathtools}
\usepackage[all]{xy}
\usepackage{setspace}
\usepackage{hyperref}
\usepackage{tikz}
	\usetikzlibrary{shapes}
	\usetikzlibrary{decorations.pathreplacing}
	\usetikzlibrary{positioning}
	\usetikzlibrary{calc}
\usepackage{adjustbox}
\usepackage{amsaddr}
\usepackage{enumitem}

\def\B{\mathcal{B}}
\def\C{\mathcal{C}}
\def\F{\mathbb{F}}

\def\Z{\mathbb{Z}}

\renewcommand{\th}{$^{\rm th}$}
\newcommand{\balpha}{\boldsymbol{\alpha}}

\renewcommand\epsilon{\varepsilon}
\renewcommand\phi{\varphi}

\let\Tr\relax
\DeclareMathOperator{\Tr}{Tr}
\let\tr\relax
\DeclareMathOperator{\tr}{tr}
\DeclareMathOperator{\RS}{RS}
\DeclareMathOperator{\GRS}{GRS}

\DeclareMathOperator{\rank}{rank}

\DeclareMathOperator{\galring}{GR}
\newcommand{\gr}[3]{\galring({#1}^{#2},{#3})}

\newtheorem{thm}{Theorem}[section]
\newtheorem{prop}[thm]{Proposition}
\newtheorem{cor}[thm]{Corollary}
\newtheorem{lemma}[thm]{Lemma}

\theoremstyle{definition}
\newtheorem{ex}[thm]{Example}
\newtheorem{defn}[thm]{Definition}
\newtheorem{rmk}[thm]{Remark}

\theoremstyle{remark}

\begin{document}

\title{Linear exact repair schemes for free MDS and Reed-Solomon codes over Galois rings}


\author[Bossaller]{Daniel P. Bossaller}
\email[UAH]{daniel.bossaller@uah.edu}
\address[UAH]{Department of Mathematical Sciences\\
University of Alabama in Huntsville\\
Huntsville, AL, USA}

\author[Lopez]{Hiram H. L\'opez}
\email[VT]{hhlopez@vt.edu}
\address[VT]{Department of Mathematics\\
Virginia Tech\\
Blacksburg, VA, USA
\vskip 1 cm
Dedicated to Professor Sudhir Ghorpade on the occasion of his sixtieth birthday.
\vskip 1 cm
}

\thanks{Hiram H. L\'opez was partially supported by the NSF grants DMS-2201094 and DMS-2401558.}

\begin{abstract}
Codes over rings, especially over Galois rings, have been extensively studied for nearly three decades due to their similarity to linear codes over finite fields. A distributed storage system uses a linear code to encode a large file across several nodes. If one of the nodes fails, a linear exact repair scheme efficiently recovers the failed node by accessing and downloading data from the rest of the servers of the storage system. In this article, we develop a linear repair scheme for free maximum distance separable codes, which coincide with free maximum distance with respect to the rank codes over Galois rings. In particular, we give a linear repair scheme for full-length Reed-Solomon codes over a Galois ring.
\end{abstract}

\maketitle

\section{Introduction}
A distributed storage system encodes a large file across several nodes or servers. If one of the nodes fails, the storage system should recover the information from the rest of the servers. The problem of recovering the failed node is known as the exact repair problem.

%

A Reed-Solomon code of length $n$ and dimension $k$ over the finite field $\F_{q^m}$ is the $k$-dimensional subspace of $\F_{q^m}^n$ formed by evaluating all polynomials of degree less than $k$ on a subset of $\F_{q^m}$ with $n$ elements. Since Reed-Solomon codes are MDS, which means maximum distance separable (i.e., their minimum Hamming distance is $d = n - k + 1$), they have the property that any $k$ entries of a codeword can be used to determine the other entries. Thus, when a Reed-Solomon code is used in a distributed storage system, a naive algorithm will download elements of $\F_{q^m}$ from $k$ different nodes. Since $\F_{q^m}$ is an $m$-dimensional vector space over $\F_q$, every element of $\F_{q^m}$ is uniquely represented by a vector in $\F_q^m$ once a basis is fixed. Therefore, we say that a naive algorithm has a repair bandwidth of $km$ over $\F_q$ because the algorithm requires $k$ times $m$ elements of $\F_q$ to solve the exact repair problem. If the elements of $\F_{q^m}$ are called symbols, and the elements of $\F_q$ are called subsymbols, the naive algorithm has a repair bandwidth of $k$ symbols or $km$ subsymbols.

The fundamental work by Guruswami and Wootters~\cite{gw17} initiated the development of linear repair schemes. These schemes efficiently solve the exact repair problem when a distributed storage system encodes the information with a Reed-Solomon
code. Guruswami and Wootters based their linear exact repair scheme on the trace function from $\F_{q^m}$ to $\F_q$, and the dual of the Reed-Solomon code, which is a (generalized) Reed-Solomon code. In contrast to the naive algorithm of a distributed storage system with a Reed-Solomon code that has to access and download $\F_{q^m}$ symbols from $k$ different nodes with a repair bandwidth of $km$ subsymbols, the linear exact repair scheme from Guruswami and Wootters needs to access and download $\F_q$ subsymbols from the remaining $n-1$ nodes. Thus, the linear exact repair scheme has a repair bandwidth of $n-1$ subsymbols.

Linear exact repair schemes have been extended to other families of linear codes. Some examples are the following. In~\cite{repairRM}, the authors developed a linear exact repair scheme for Reed-Muller codes, which are defined by the evaluation of multivariate polynomials up to a certain degree over the set of all points in $\F_{q^m}^s$. In~\cite{repairAG}, the authors developed a linear exact repair scheme for algebraic geometry (AG) codes, which are defined, in a few words, by the evaluation of rational functions over the set of rational points of a curve. In~\cite{repairaugmented}, the authors developed a linear exact repair scheme for augmented Reed-Muller codes, which are defined by strategically adding the evaluations of extra monomials to a Reed-Muller code.

Linear codes over finite rings, i.e., submodules of $R^n$, where $R$ is a ring other than a finite field, have been extensively studied over the last thirty years. 
In~\cite{HKCS1994}, the authors showed that the well-known nonlinear binary Kerdock and Preparata codes are linear codes over the ring $\Z_4$ under an appropriate map. In~\cite{HLCodesChain2000}, the authors studied codes over chain rings, whose ideal lattice is linearly ordered. Linear codes over principal ideal rings and Galois rings were studied in~\cite{PIRcodes} and~\cite{galoiscodes}, respectively. Chain rings, principal ideal rings, and Galois rings are all examples of (finite) Frobenius rings. In two groundbreaking results, Wood showed in~\cite{wood99} and \cite{wood08} that finite Frobenius rings are particularly appropriate for linear codes over rings because the classical MacWilliams Identities and the MacWilliams Extension Theorem hold for linear codes over a finite ring if and only if the finite ring is a finite Frobenius ring; these two crucial results in coding theory were established initially in Jesse MacWilliams's Dissertation~\cite{FJMacWilliamsThesis}. Families of codes over finite rings, such as cyclic, quasi-cyclic, and convolutional codes, have been studied in~\cite{codesandrings} and~\cite{doughertybook}. Various classes of linear codes over finite fields, like Reed-Muller codes and generalized Reed-Solomon codes, have their counterpart in the finite ring setting. In~\cite{QBC13}, the authors introduce generalized Reed-Solomon codes over finite commutative rings and generalize their classical Welch-Berlekamp and Guruswami-Sudan algorithms for maximum likelihood decoding and list decoding. In this paper, we develop a linear exact repair scheme when a distributed storage system encodes the information using a free code over a Galois ring. Galois rings are widely studied due to their structural similarities to finite fields. While it is common for codes over finite rings to consider different metrics, we focus in this article on the Hamming metric because of the present article's focus on distributed storage systems.

This article is organized as follows. In Section~\ref{RSoverRings}, we review the theory of Galois rings. We define free MDS codes over Galois rings and then the Teichm\"uller set of a Galois ring to introduce generalized Reed-Solomon codes over a Galois ring. In Section~\ref{MDScodes}, we review the theory of regenerating codes and linear exact repair schemes for Reed-Solomon codes over finite fields, focusing mainly on the Guruswami-Wootters repair scheme. In Section~\ref{RepairGalois}, we extend the theory of regenerating codes and linear exact repair schemes to free MDS codes over Galois rings. In Section~\ref{RSGalois}, we define a linear exact repair scheme for full-length Reed-Solomon codes over a Galois ring. In Section~\ref{conclusion}, we comment on future directions for codes over more general Frobenius rings.

\section{Preliminaries}\label{RSoverRings}
In this section, we introduce Galois rings and list many of their properties that will be useful in the sequel. We then describe maximum distance separable (MDS) and generalized Reed-Solomon codes over rings and their main properties.

\subsection{Galois rings}\label{24.05.25}


Let $\Z_{p^n}$ be the ring of integers modulo a prime power $p^n$. As this ring $\Z_{p^n}$ is local with maximal ideal generated by $\langle p \rangle = p \Z_{p^n}$ and residue field $\Z_{p^n} / \langle p \rangle \simeq \F_p$, the canonical surjective map $\bar \cdot: \Z_{p^n} \to \F_p$ extends naturally to a surjective map on the polynomial rings
\[\bar \cdot: \Z_{p^n}[x] \to \F_{p}[x].\]
An element $f(x) \in \Z_{p^n}[x]$ is called a {\bf monic basic irreducible} polynomial if $\overline{f(x)} \in \F_p[x]$ is monic and irreducible over the prime field $\F_p$.

\begin{defn}
Let $f(x) \in \Z_{p^n}[x]$ be a monic basic irreducible polynomial of degree $m$ with coefficients from $\Z_{p^n}$. The ring
\[\gr pnm := \frac{\Z_{p^n}[x]}{\langle f(x)\rangle}\]
is called the {\bf Galois ring} of degree $m$ over $\Z_{p^n}$. The Galois ring $\gr pnm$ has characteristic $p^n$; thus, we will call $\Z_{p^n}$ the {\bf characteristic ring} of $\gr{p}{n}{m}$.
\end{defn}
Observe that the Galois ring of degree $m$ over the finite field $\Z_p \simeq \F_p$ is the finite field $\F_{p^m}$. On the other hand, the Galois ring of degree $1$ over $\Z_{p^n}$ is precisely $\Z_{p^n}$. In the following we summarize a series of properties about Galois rings.

\begin{prop} \label{GaloisRingproperties}(\cite{mcdonald}, Ch. XVI and \cite{lecffgr}, Ch. 14)
For a Galois ring $\gr pnm$, the following properties hold.
\begin{enumerate}[label={\rm (\alph*)}]
    \item $\gr pnm$ has a unique maximal ideal $\langle p \rangle = p \gr pnm$, and its residue field is $\gr pnm / \langle p \rangle \simeq \F_{p^m}$.
    
    \item The ideals of $\gr pnm$ form a chain
    \[\gr pnm \supseteq \langle p \rangle \supseteq \langle p^2 \rangle \supseteq \cdots \supseteq \langle p^{n-1}\rangle \supseteq \langle p^n \rangle= 0.\]
    
    \item For each $i$, $\langle p^i \rangle / \langle p^{i+1} \rangle$ is a one dimensional vector space over $\F_{p^m}$.

    \item Every Galois ring $\gr{p}{n}{ml}$, where $l$ is an integer $\geq 1$, is a Galois extension of $\gr pnm$. That is, there is a monic basic irreducible polynomial $\rho(x) \in \gr pnm[x]$ so that $\deg(\rho(x)) = l$ and 
    \[\gr{p}{n}{ml} \simeq \frac{\gr{p}{n}{m}[x]}{\langle \rho(x) \rangle}.\]
    \item Given a Galois ring $\gr{p}{n}{ml}$, where $l$ is an integeger $\geq 1$, the group of automorphisms of $\gr{p}{n}{ml}$ which fix $\gr pnm$ is a cyclic group of order $l$ generated by the map $\sigma_m: \gr{p}{n}{ml} \to \gr{p}{n}{ml}$ given by $\sigma_m(x) = x^{p^m}$.
\end{enumerate}
\end{prop}

Let $\bar \cdot : \gr pnm \to \F_{q^m}$ be the canonical projection map and choose some $\gamma \in \gr pnm$ so that $\bar \gamma$ is a primitive element of the residue field $\F_{p^m}$. Then $\gamma$ is a unit of $\gr pnm$ of order $p^m - 1$. Furthermore, it follows from Proposition \ref{GaloisRingproperties}, parts (b) and (c) that every $a \in \gr pnm$ has a unique $p$-adic representation
\[a = a_0 + a_1 p + \cdots + a_{n-1}p^{n-1},\]
where $a_i \in \mathcal T := \{0\}\cup\{\gamma^i : 0 \leq i \leq p^m-1\}$.
The set $\mathcal T$, which consists of coset representatives of the distinct elements of the field $\F_{p^m}$, is called the {\bf Teichm\"uller set} of $\gr pnm$.

\begin{rmk}
    Note that the set $\mathcal T$ is not unique; given another $\beta \in \gr pnm$ such that $\bar \beta$ is primitive, we might construct a different Teichm\"uller set by including powers of $\beta$ rather than powers of $\gamma$. For this reason, when working with a Galois ring, we will implicitly fix some Teichm\"uller set $\mathcal T$.
\end{rmk}
\subsection{Codes over rings}
In this subsection, we introduce some of the main objects in this manuscript: free maximum distance separable (MDS) and generalized Reed-Solomon codes over commutative rings. Throughout the article, unless otherwise stated, $R$ will be assumed to be a commutative ring with identity. In particular, $R$ can be assumed to be a Galois ring.

We start with the general definition of codes over rings.
Let $R$ be a commutative ring with identity.
\begin{defn}
A (linear) {\bf code} of length $n$ over the ring $R$ is an $R$-module $\C \subseteq R^n$. 
The rank of $\C$, denoted by $\rank(\C)$, is defined as the minimum number of generators of $\C$ as an $R$-module. We denote by $d(\C)$ the minimum Hamming weight of the code $\C$, i.e.
\[d(\C) := \min\{\operatorname{wt}(c) : c \in \C, c\neq 0 \},\]
where $\operatorname{wt}(c)$ represents the number of non-zero entries of $c$. The code $\C$ is a {\bf free code} if $\C$ is a free module over $R$.
\end{defn}
\begin{rmk}
Note that codes over rings do not need to as clear-cut a definition of `dimension' as is the case for codes over fields. Indeed, for an $R$-module $\C$, the size of the smallest generating set for $\C$ need not be the same as the rank of the largest free submodule of $\C$.
\end{rmk}

Let $\C$ be a linear code over $R$. The well-known Singleton bound (see ~\cite{bok:MW}) states
\[d(\C) \leq n - \log_{|R|}|\C| +1.\]
A code meeting this bound is called a {\bf maximum distance separable (MDS)} code. There is also a bound in terms of the rank. By~\cite{Dougherty}, we have
\[d(\C) \leq n - \rank(\C) +1.\]
A code meeting the previous bound is called a {\bf maximum distance respect to Rank (MDR)} code.

Assume that $\C \subseteq R^n$ is free of rank $k$. As $\C$ is isomorphic to $R^k$ as an $R$-module, then
\[\log_{|R|}(|\C|) = k = \rank(\C).\]
\begin{rmk}\label{24.05.26}
The previous equation implies that for a free code $\C \subseteq R^n$, the property of being MDS is equivalent to being MDR. From now on, we will mainly use the term MDS for a free code that meets the Singleton bound.
\end{rmk}

Following the lead of Quinton, Barbier, and Chabot in~\cite{QBC13}, we define and give basic properties of generalized Reed-Solomon codes over a ring. Reed-Solomon codes are defined by evaluating polynomials on certain points called evaluation points. Reed-Solomon codes over finite fields benefit from the fact that for any two evaluation points $a_i \neq a_j$ in the finite field $\F_q$, the difference $a_i - a_j$ is nonzero and invertible. This property certainly is not generally true for rings, which motivates the following definition.

\begin{defn}\label{25.05.18}
Let $R$ be any ring and $U(R)$ the group of units of $R$. A subset $L \subseteq R$ is called {\bf subtractive} if for every pair of distinct elements $a$ and $b$ in $L$, the difference $a-b$ is an element of $U(R)$.
\end{defn}
Given two elements $a,b$ in a subtractive set $L$, Definition~\ref{25.05.18} requires the difference $a-b$ to be invertible, but $a-b$ does not need to be in $L$. For a finite field $\mathbb{F}_q$, any subset $L \subseteq \mathbb{F}_q$ is subtractive,
and the largest subtractive set is $L = \mathbb{F}_q$. In a ring $R$, the size of the largest $L$ depends on the structure of the ring. For example, in $\mathbb{Z}_4$, $L= \{0, 1\}$, $L= \{0, 3\}$, $L= \{1, 2\}$, and
$L= \{2, 3\}$ are subtractive sets, and no subtractive set can be constructed using three elements from $\mathbb{Z}_4$. In $\mathbb{Z}_{15}$, $L=\{0,2,4\}$ is a subtractive set with three elements. In general, we can ensure that every ring with identity has a subtractive subset of size at least two, namely $L = \{0, 1\}$.

\begin{defn}
Let $R$ be a ring and $U(R)$ the group of units of $R$.
For a subtractive subset $\balpha = \{\alpha_1, \alpha_2, \ldots, \alpha_n\}$, and a vector of units $\boldsymbol{v} = (v_1, v_2, \ldots, v_n) \in U(R)^n$, the left $R$-module \[\RS_{\boldsymbol{v}}(n,k) = \{(v_1 f(\alpha_1), v_2 f(\alpha_2), \ldots, v_n f(\alpha_n)): f(x) \in R[x] \text{ and } \deg(f) < k\}\]
is called the {\bf generalized Reed-Solomon code} of length $n$ and pseudo-dimension $k$ over $R$. The vector $\boldsymbol{v}$ is called the vector of multipliers. When $\boldsymbol{v} = (1, 1, \ldots, 1)$, $\RS_{\boldsymbol{v}}(n,k)$ is denoted by $\RS(n,k)$.
\end{defn}

We have the following properties for generalized Reed-Solomon codes over a commutative ring with identity.
\begin{prop}(\cite{QBC13})\label{RSMDS}
Let $\mathcal C = \RS_{\boldsymbol{v}}(n, k)$ be a generalized Reed-Solomon code over a commutative ring with identity. Then, the following holds.
\begin{enumerate}[label={\rm (\alph*)}]
\item $\mathcal C$ is a free code over $R$.
\item As a left module, $\rank(\C) = k$.
\item Under the standard Euclidean inner product, the dual of a generalized Reed-Solomon code is again a generalized Reed-Solomon code, with the same evaluation set $\balpha$, but a different vector of multipliers. In other words,
\[\C^\perp = \RS_{\boldsymbol{v}'}(n,n-k),\]
for some ${\boldsymbol{v}'} \in U(R)^n$.
\item $\mathcal C$ is MDS, that is, $d(\mathcal C) = n-k+1$.
\end{enumerate}
\end{prop}

In~\cite{QBC13}, Quintin, Barbier, and Chabot note that the assumption of the subtractivity of $\balpha$ is necessary for many of the properties of polynomials over rings, such as the correspondence between the roots $f(a) = 0$ and the factors $(x - a)$ of $f(x)$. By Remark~\ref{24.05.26}, Reed-Solomon codes are also MDR, maximum distance with respect to rank, because they are free by Proposition~\ref{RSMDS}~(a).

\section{Repair schemes for codes over finite fields}\label{MDScodes}
Before introducing repair schemes for generalized Reed-Solomon codes over Galois rings, we give in this section the theory of regenerating codes and repair schemes in the context of finite fields.

A distributed storage system uses a linear code to encode a large file across $n$ nodes or servers. One of the main goals of a distributed storage system is to reliably store files, which is paramount, so if one of the servers fails, the storage system introduces a recovery server to replace the failed server.


Maximum distance separable codes are desirable for distributed storage systems because they achieve optimal tradeoff between the dimension of the code and minimum distance; given an $[n,k]$ MDS code, any $k$ symbols uniquely determine the codeword. But this tradeoff is sharp in that no codeword is determined by fewer than $k$ symbols. Indeed, consider that a distributed storage system utilizes an $[n,k]$ MDS code to store the information in $n$ nodes $f_1,\ldots,f_n$. Then, if there is an erasure at the node $f_1$, the storage system introduces a repair node $f_r$ to replace the failed node $f_1$. Then, the storage system must consult at least $k$ distinct helper nodes to repair the single lost node. As the storage system is using a $[n,k]$ MDS code, the inefficiency is clear since for $f_r$ to replicate the information at node $f_1$, the $k$ helper nodes must communicate a total of $k$ times the amount of data stored in the lost node.

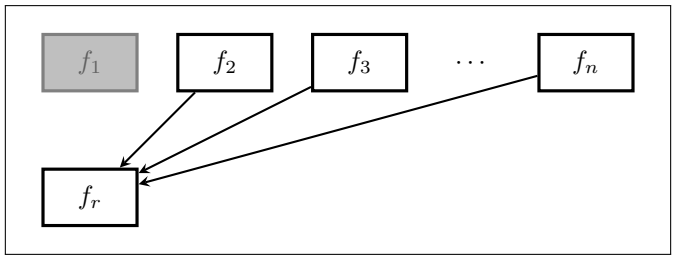
\begin{center}
\fbox{\adjustbox{padding = .25cm}{
\begin{tikzpicture}[node distance = 1cm and .5cm]
\tikzset{nd/.style={rectangle, draw = black, fill = white, very thick, minimum height = .75cm, minimum width = 1.25cm}}
\tikzset{fnd/.style={rectangle,draw = black, fill = gray, opacity = .5, very thick, minimum height = .75cm, minimum width = 1.25cm}}
\tikzset{rnd/.style={rectangle, draw = black, fill = green, very thick, minimum height = .75cm, minimum width = 1.25cm}}
\node[fnd] (f1) at (0,0) {$f_1$};

\node[nd, right = of f1] (f2){$f_2$};
\node[nd, right = of f2] (f3){$f_3$};
\node[right = of f3](cdots){$\cdots$};
\node[nd, right = of cdots] (fn){$f_n$};

\node[nd, below = of f1] (r){$f_r$};

\draw[-stealth, thick] (f2) to (r);
\draw[-stealth, thick] (f3) to (r);
\draw[-stealth, thick] (fn) to (r);

\end{tikzpicture}
}}
\end{center}

Regenerating codes were introduced by Dimakis et al. in \cite{dimakisetal} to address the inefficiency of MDS codes in distributed storage systems. In essence, regenerating codes may query more than $k$ helper nodes, but they only download part of the information contained at the node. In linear codes over $\F_{q^m}$, this is accomplished through ``subpacketization" via subsymbols, that is, noting that a single symbol of $\F_{q^m}$ may be regarded as a vector in the $\F_{q}$-vector space $\F_q^m$. The helper node may, instead, send a vector of subsymbols from $\F_q$.

Let $\F_{q} \subseteq \F_{q^m}$ be an extension of finite fields. The {\bf (field) trace} of $\F_{q^m}$ over $\F_{q}$ is a surjective $\F_q$-linear map $\tr:\F_{q^m} \to \F_{q}$ defined as \[\tr(x) = \sum_{i = 0}^{m-1} x^{q^i}.\]

\begin{rmk}\label{autofields}
Since every automorphism of $\F_{q^m}$ that fixes $\F_q$ is of the form $\sigma_i(x):= x^{q^i}$, we can equivalently define the field trace as the sum of the distinct automorphisms of $\F_{q^m}$ that fix $\F_q$.
\end{rmk}

The field trace $\tr: \F_{q^m} \to \F_q$ proves to be very effective for subpacketization since for any basis $\B = \{\beta_i: 1 \leq i \leq m\}$ for $\F_{q^m}$ over $\F_q$, there exists a trace-dual basis $\B^* = \{\beta_i^*: 1 \leq i \leq m\}$ such that for any $a \in \F_{q^m}$,
\[a = \sum_{i = 1}^m \tr(\beta_i^* a) \beta_i.\] In other words, the coefficients for the expansion of $a$ with respect to the basis $\B$ are calculated using the trace-dual basis. See \cite{LNffields} Chapter 2 for a deeper theoretical background on the trace and the trace-dual basis for finite fields.

Given a file encoded using the code $\C$ over a finite field $\F_{q^m}$, which has a dual code $\C^\perp$, Guruswami and Wootters define a linear repair scheme for the lost node $f_i$ \cite{gw17} in the following way. Given the codeword $f = (f_1, f_2, \ldots, f_n)$, there exists some codeword $g \in \C^\perp$ such that $f \cdot g = 0$ (recall that `$\cdot$' denotes the standard Euclidean inner product). Thus, we may solve for $g_i f_i$ to obtain
\[g_i f_i = - \sum_{j \neq i} g_j f_j.\]
If there is a family of codewords $\mathcal G = \{g^1, g^2, \ldots, g^m\}$ such that the $i$\th coordinate $\{g_i^1, g_i^2, \ldots, g_i^m\}$ form a basis for $\F_{q^m}$ over $\F_q$, then the trace-dual basis $\{(g_i^j)^* : 1 \leq j \leq m\}$ recovers the lost node through the following equations
\[f_i = \sum_{j = 1}^m \tr(g_i^j f_i) (g_i^j)^*.\]

In practice, the repair depends upon the following procedure. Upon being contacted by the repair node, node j will calculate and send a vector \[(\tr(g_j^1 f_j), \tr(g_j^2 f_j), \ldots \tr(g_j^m f_j)) \in \F_q^m.\]
The advantage of this scheme is that the choice of the collection of dual codewords $\mathcal G \subseteq \C^\perp$ has only one constraint in that $\{g_i^1, g_i^2, \ldots, g_i^m\}$ must form a basis for $\F_{q^m}$ over $\F_q$. With this requirement fulfilled, the MDS property of $\C$ can be used to pick the specific family of codewords, $\mathcal G$, which is of full rank at the $i$\th coordinate, and of low rank in the other coordinates. In other words, $\rank(\{g_i^j : 1 \leq j \leq m\}) = m$, but $\rank(\{g_l^j : l \neq i \text{ and } 1 \leq j \leq m \}) \leq m$. 

The assumption that the code $\C$ is MDS provides the maximum amount of flexibility when designing repair schemes for $\C$. Since any $k$ coordinates uniquely determine the codeword, choose the codewords in $\mathcal G$ which have a basis $\{g_i^j : 1 \leq j \leq m\}$ at coordinate $i$, leaving $k - 1$ coordinates which may be freely chosen. Specific choices for values in these remaining coordinates can be used to reduce the rank of the sets $\{g_{i'}^j : 1 \leq j \leq m \text{ and } i' \neq i\}$. Using this technique, Guruswami and Wootters provide an optimal repair scheme for Reed-Solomon codes $\RS(n,k)$ where $k \leq n(1 - 1/q)$, and a repair scheme for a $\GRS(14,10)$ generalized Reed-Solomon code which improved on the state-of-the-art at the time. This specific generalized Reed-Solomon code was implemented in the Apache Software Foundation Hadoop Distributed File System \cite{HDFS}.

\section{Repair schemes for codes over Galois rings}\label{RepairGalois}
Before examining specific instances of a repair scheme for Reed-Solomon codes over Galois rings, we explore recovery algorithms for codes over Galois rings in general. To this aim, we first define the trace between Galois rings, which by Proposition~\ref{GaloisRingproperties}~(e), generalizes the idea that the field trace $\tr: \F_{q^m} \to \F_q$ is defined as the sum of all the distinct field automorphisms of $\F_{q^m}$ that fixes $\F_q$ (see Remark~\ref{autofields}). We refer the reader to~\cite{GNGF} for a more general treatment of the trace function over other classes of finite rings.
\begin{defn}\label{24.05.23}
Let $S = \gr{p}{n}{ml}$ and $R = \gr{p}{n}{m}$ be Galois rings. The {\bf trace} of $S$ over $R$ is the $R$-linear map $\Tr_R^S: S \rightarrow R$ defined as
\[\Tr_{R}^S(x) = \sum_{i=0}^{l-1} \sigma_m^i(x),\]
where $\sigma_m(x) = x^{p^m}$ is the generator of the cyclic group of automorphisms of $S$ that fix $R$.
\end{defn}
Note that when $S = \gr{p}{n}{m}$ and $R = \gr{p}{n}{1} = \Z_{p^n}$, we have
\[\Tr_{R}^S(x) = \sum_{i=0}^{m-1} \sigma_m^i(x) = \sum_{i=0}^{m-1} x^{p^{i}}.\] We will denote the above trace map from a Galois ring onto its characteristic ring simply by `$\Tr(x)$'.

The following result is a simple computation.
\begin{lemma}
Let $a = \sum_{i = 0}^{n-1} a_ip^i$ be the $p$-adic representation of some element $a $ in $\gr{p}{n}{m}$. Then,

\[\Tr(a) = \sum_{i=0}^{n-1} \tr(a_i) p^i,\]
where $\tr(x)$ denotes the field trace $\tr: \F_{p^m} \rightarrow \F_{p}$ applied to the Teichm\"uller elements.
\end{lemma}
\begin{proof}
As $\Tr(x)$ is a $\Z_{p^n}$-linear map and the $p^i$ are elements of $\Z_{p^n}$, we have
\[\Tr(a) =
\Tr \left(\sum_{i = 0}^{n-1} a_ip^i\right)=
\sum_{i = 0}^{n-1} \Tr \left(a_i\right)p^i.\]
By Definition~\ref{24.05.23}, we have that $\Tr(x) = \tr(x)$. Thus, we obtain the result.
\end{proof}

\begin{rmk}
Recall that the Teichm\"uller set is a collection of coset representatives of the ideal $\langle p \rangle$ in $\gr pnm$. Since we assume $\mathcal T$ is fixed, there is no ambiguity in the images of the field trace $\tr$ on the elements of $\mathcal T$.

\end{rmk}


From now on, we fix the rings $R:= \gr pnm$ and $S:= \gr pn{lm}$. Naturally $R \subseteq S$, and both are algebraic extensions of $\Z_{p^n}$ via polynomials $f$ and $g$ for $R$ and $S$ respectively. Note that $\bar f$ and $\bar g$ are polynomials which generate the fields $\F_{p^m}$ and $\F_{p^{ml}}$, and hence $\bar f \mid \bar g$ in $\Z_p[x]$.

The next lemma shows that given two Galois rings $S \supseteq R$, when we fix a Teichm\"uller set for $S$, we may find a subset which is a Teichm\"uller set for the subring $R$. 
\begin{lemma}\label{teichsubset}
Let $R= \gr pnm$ and $S= \gr pn{lm}$ be Galois rings and fix a Teichm\"uller set $\mathcal T$ of $S$.
    \begin{enumerate}[label={\rm (\alph*)}]
    \item There is some $\alpha \in \mathcal T$ such that $\{0\}\cup\{\alpha^i : 0 \leq i \leq p^m - 2\}$ is a Teichm\"uller set for $R$.
    \item There is some $\beta \in \mathcal T$ such that $\mathcal B = \{\beta^j : 0 \leq j \leq l-1\}$ is a basis for $S$ as a free module over $R$.
    \end{enumerate}
\end{lemma}
\begin{proof}
(a) Let $\gamma$ be an element in $S$ such that $\mathcal T := \{0\}\cup\{\gamma^i : 0 \leq i \leq p^{ml}-1\}$ and $\alpha \in \mathcal T$ given by
\[\alpha = \gamma^{\frac{p^{ml} - 1}{p^m - 1}}=\gamma^{p^{m(l-1)} + p^{m(l-2)} + \cdots + p^m +1}.\]
Then, $\sigma_m(\alpha) = \alpha^{p^m}=\alpha$, so $\alpha \in R$. Moreover, it is clear that the order of $\alpha$ is $p^m-1$ in $U(R)$ and the projection of $\alpha$ onto the residue field of $R$ is a primitive element. This proves (a).


(b) Since $S$ is a Galois extension of $R$, by Lemma \ref{GaloisRingproperties} (d), there is a monic basic irreducible polynomial $\rho(x)$ of degree $l$ such that $S \simeq R[x]/\langle\rho(x)\rangle$. Under the mapping down to the residue field of $S$, $\F_{p^m}$, we have that the following
\[\bar \rho(x) \mid (x^{p^{ml}} - x) \text{, hence } \bar \rho(x) \mid \left(x^{p^{m}} - x\right) \cdot \frac{x^{p^{ml} - 1} - 1}{x^{p^{m} - 1} - 1}.\]
Note that $\frac{x^{p^{ml} - 1} - 1}{x^{p^{m} - 1} - 1}$ is a polynomial. 

Since $\bar \rho(x)$ is irreducible over $\F_{p^m}$, $\bar \rho(x)$ cannot divide $x^{p^m} - x$ in $\F_{p^m}[x]$, hence $\bar \rho(x)$ must divide $\frac{x^{p^{ml} - 1} - 1}{x^{p^{m} - 1} - 1}$. Thus, by Hensel's Lemma (\cite{mcdonald}, Theorem XIII.4), there must exist some polynomial $h(x) \in R[x]$ such that \[\rho(x) \cdot h(x) =  \frac{x^{p^{ml} - 1} - 1}{x^{p^{m} - 1} - 1}.\] Choose some $\beta \in  S$, such that $\bar \beta$ is a zero of $\bar \rho(x)$ and $\bar \beta$ is a primitive element of $\F_{p^m}$. Then it's evident that $\beta^i \in \mathcal T$. Therefore, $\B = \{1, \beta, \beta^2, \ldots, \beta^{l-1}\}$ forms a basis for $S$ as a free module over $R$ since $\beta$ is a zero of $\rho(x)$. Furthermore, $\mathcal B \subseteq \mathcal T$ by construction.
\end{proof}

Lemma~\ref{teichsubset} (b) gives us the tools to define subpacketization for codes over rings. Indeed, since $S$ is a free $R$ module, we may write any $s \in S$ uniquely as an $R$-linear combination of $\mathcal B$ 
\[s = \sum_{i} a_i \beta^i.\]

For the Galois ring $R = \gr pnm$, there exists a basis $\B = \{b_1, b_2, \ldots, b_m\}$ for $R$ as a free module over its characteristic ring $\Z_{p^n}$. Furthermore, the following results state that a dual basis exists for $B$ with respect to the trace.
\begin{prop}\cite[Theorem~3.7]{sison20}\label{24.05.24}
Let $\B = \{b_1, b_2, \ldots, b_m\}$ be a basis for the Galois ring $R = \gr pnm$ as a free module over its characteristic ring $\Z_{p^n}$. There exists a unique dual basis $\B^* = \{b_1^*, b_2^*, \ldots, b_m^*\}$ such that
\[\Tr_{\Z_{p^n}}^R(b_i^* b_j) = \delta_{ij}.\]
\end{prop}

\begin{rmk}
It should be noted that if $S$ is non-commutative, then for a basis $\B$ which generates $S$ as a free left $R$-module, $_{R}S$, the dual basis $\B^*$ is a basis for the free {\it right} $R$-module, $S_{R}$. This distinction is unnecessary for the present work because we assume that the rings are commutative.
\end{rmk}

Let $\mathcal C \subseteq R^n$ be a free MDS code over the Galois ring $R = \gr pnm$, with length $n$, rank $k$, and (due to meeting the Singleton bound) minimum distance $d = n - k + 1$. Since this code $\C$ is a free code over a principal ideal ring, $\C$ has an MDS dual code $\C^{\perp}$ of length $n$ and rank $n - k$~\cite{CodesoverChain}.

\begin{rmk}
The assumption that $\mathcal C$ is MDS implicitly assumes that the residue field of the finite local ring $R$ with maximal ideal $\mathfrak{m}$ is sufficiently large to admit the MDS property. In~\cite{Dougherty}, Dougherty, Kim, and Kulosman show that the assumption that
\[|{R/\mathfrak{m}}| > \binom{n-1}{n-k-1}\] is sufficient.
\end{rmk}

As a consequence of Proposition~\ref{24.05.24}, we note that dual bases for Galois rings retain many of the same properties as their finite field counterparts. In particular, for every element $a$ in $R = \gr pnm$, the following decomposition holds
\[a = \sum_i \Tr_{\Z_{p^n}}^R(b_i^* a) b_i.\]

By construction of $S= \gr pn{lm}$ as a Galois extension of $R$, and due to Lemma \ref{teichsubset}, there is a subset of the Teichm\"uller set of $S$ which is a basis for $S$ as a free module over $R$, and a corresponding dual basis with respect to the trace function $\Tr_R^S: S \to R$. This leads to the following result, which mirrors Proposition 6 and Corollary 7 in \cite{gw17}.

\begin{cor}
Let $R= \gr pnm$ and $S= \gr pn{lm}$ be Galois rings and $\B = \{b_1, b_2, \ldots, b_m\}$ a basis for $S$ as a free module over $R$. Then, for any $a, \zeta \in S$, there exist unique scalars $\mu_{i,\zeta}$ from $R$ such that \[\zeta a = \sum_{i}\mu_{i,\zeta} b_i.\]
\end{cor}
\begin{proof}
By Lemma~\ref{teichsubset}~(b), $S$ may be considered to be a free $R$-module; existence and uniqueness follow quickly from freeness since $a\zeta \in S$. 
\end{proof}

We come to one of the main results of this section, which ensures the existence of a repair scheme over a Galois ring.
\begin{thm}\label{linear repair}
Any free MDS $[n,k]$ code over the Galois ring $S= \gr pn{lm}$ admits a linear exact repair scheme over the Galois ring $R= \gr pnm$.
\end{thm}

\begin{proof}
Let $Z \subseteq S$ be a set that freely generates $S$ as an $R$-module. Suppose we have lost the $i$\th coordinate $c_i^*$ of a codeword $c \in \C$. Since $\C$ is a free MDS code, its dual code is a free $[n,n-k]$ MDS code $\C^\perp$. Let $g \in \C^\perp$ so that $g \cdot c = 0$, then we may write the unknown coordinate
\[g_i c_i^* = -\sum_{j \neq i} g_jc_j.\]
Taking the trace of both sides and using its linearity gives 
\[\tr(g_ic_i^*) =- \sum_{j \neq i} \tr(g_j c_j).\]
Since $\C^\perp$ is MDS, it is possible to pick $m$ codewords $g^1, g^2, \ldots, g^m \in \C^\perp$ such that $\{g_i^j : 1 \leq j \leq m\} = Z$, that is, the $i$\th coordinates of these codewords forms a set which generates $R$ as a $\Z_{p^n}$-module. Then there exists a trace-dual basis $\{\nu_i^j : 1 \leq j \leq m\}$. such that $\tr(g_i^j \nu_{i}^{j'}) = \delta_{jj'}$.

Thus for each $j$ the repair node may reconstruct each value $\tr(g_i^j c_i^*)$; these trace values can then be used to find 
\[c_i^* = \sum_{i=1}^n \tr(g_i^j c_i^*) \nu_i^j.\]
Thus, we have the desired result.
\end{proof}

\section{Repair schemes for Reed-Solomon codes over a Galois Ring}\label{RSGalois}
In this section, we define a repair scheme for a Reed-Solomon code over a Galois ring.

Let $R= \gr pnm$ and $S= \gr pn{lm}$ be Galois rings. The Teichm\"uller set of the Galois ring $S$ is denoted by $\mathcal{T}=\{0\} \cup \{\gamma^{i - 1} : 1 \leq i \leq p^{ml} -1\}$, (see Subsection~\ref{24.05.25}). We define the full-length Reed-Solomon code over the Galois ring $S$ by taking the evaluation points as the Teichm\"uller set $\mathcal{T}$. In other words, the {\bf full-length} Reed-Solomon code over $S$ is denoted and defined by 
\[\RS(\mathcal{T},k) := \{(f(\alpha_0), f(\alpha_2), \ldots, f(\alpha_{p^{lm}})): f(x) \in R[x] \text{ and } \deg(f) < k\},\]
where $\mathcal{T} = \{\alpha_0, \alpha_1, \ldots, \alpha_{p^{lm} - 1}\}.$

We define $q = p^m$ so that $|R| = q^n = p^{nm}$ and $|S| = q^{nl} = p^{nml}$. As the nonzero elements of the Teichm\"uller set of $S$ form a cyclic group of order $p^{ml}-1$, the proof of the following lemma is identical to its finite field counterpart.
\begin{lemma}\label{fulldual}
Let $\mathcal C = \RS(\mathcal{T},k)$ be the full-length Reed-Solomon code over the Galois ring $S= \gr pn{lm}$.
Then, the dual of $\mathcal{C}$ is given by $\C^\perp = \RS(\mathcal{T}, n - k)$.
\end{lemma}

Lemma~\ref{fulldual} is essential for the construction of the set of linearly independent codewords $\{g_i: 1 \leq i \leq l\} \subseteq \C^\perp$ as in the proof of Theorem~\ref{linear repair}. With this in mind, suppose that a file of rank $k$ over $S$ is encoded using the full-length Reed-Solomon code $\RS(\mathcal{T},k)$, where $\mathcal{T}$ is the Teichm\"uller set of $S$. Let $f(x) \in S[x]_{< k}$ be the message polynomial associated with the file, which has a corresponding codeword
\[\begin{array}{rl}f(\mathcal{T}) = \left(f(\alpha_0), f(\alpha_1), \ldots, f(\alpha_{p^{ml} - 1} \right),\end{array}\]
which may be denoted by $f(\mathcal{T}) = \left(f_0, f_1, \ldots, f_{p^{ml} - 1}\right)$ for simplicity. Now, assume the information at node $i$, namely $f_i \in S$, is erased. Then as in Theorem \ref{linear repair}, we fix a basis $\B = \{\beta_1, \beta_1, \ldots, \beta_l\}$ for $S$ as a free $R$-module, and, since $S$ is a Galois extension of $R$, let $\widehat{\B} = \{\widehat \beta_1, \widehat \beta_2, \ldots, \widehat \beta_l\}$ be the trace dual basis for $\B$. Furthermore, for ease of notation, we will write $\Tr_{R}^{S}$ as simply $\Tr$. We introduce the following family of $l$ repair polynomials associated with $i$:
\begin{equation}\label{24.05.27}
\mathcal P(\alpha_i) = \left\{p_{i,j}(x) = \frac{\Tr(\beta_j(x - \alpha_i))}{x - \alpha_i} : 1 \leq j \leq l \right\}.
\end{equation}
Note that $p_{i,j}(x) = \sum_{a = 0}^{l-1} (\beta_j)^{q^a}(x - \alpha_i))^{q^a-1} - 1$; hence the polynomial $p_{i,j}(x)$ has degree $q^{l-1}$. Following the results of Proposition \ref{RSMDS}, so long as $\deg(p_j(x)) \leq n-k$, $p_{i,j}(x)$ is the message polynomial corresponding to some codeword in $\RS(\mathcal{T},k)^\perp = \RS(\mathcal{T},n-k)$. Furthermore $p_{i,j}(\alpha_i) = \beta_j$, while $p_{i,j}(\alpha_j) = \frac{\gamma}{\alpha_j - \alpha_i}$ for $\gamma \in R$, since $\gamma$ is the image of the $\Tr$.

We come to one of the main theorems of this section, which is a linear exact repair scheme for a full-length Reed-Solomon code over a Galois ring.
\begin{thm}
Let $R= \gr pnm$ and $S= \gr pn{lm}$ be Galois rings and $\mathcal T$ the Teichm\"uller set of $S$. If $k$ and $l$ are such that $k \leq p^{ml}(1 - \frac{1}{p^m})$, then there exists a linear exact repair scheme for the full-length Reed-Solomon code $\RS(\mathcal{T},k)$ with a repair bandwidth of $|\mathcal T|-1$ over $R$.
\end{thm}
\begin{proof}
Let $f(\mathcal{T}) = \left(f_0, f_1, \ldots, f_{p^{ml} - 1}\right)$ be an element of the full-length Reed-Solomon code $\RS(\mathcal{T},k)$.
Suppose that the information $f_i$ stored at node $i$ is erased. Let $k$ and $l$ be such that $k \leq p^{ml}(1 - \frac{1}{p^m})$. By the discussion before the theorem, every element in the family of repair polynomials given in Eq.~(\ref{24.05.27}) $\mathcal P(\alpha_i)$ defines an element in $\RS(\mathcal{T},k)^\perp = \RS(\mathcal{T},n-k)$. Then, for each $1 \leq j \leq l$, we can write the repair equation

\begin{align*}
p_{i,j}(\alpha_i) f_i &= - \sum \limits_{i' \neq i}p_{i,j}(\alpha_{i'}) f_{i'}.
\end{align*}
Hence
\begin{align*}
\beta_i f_i &= -\sum \limits_{i' \neq i} \left(\frac{\Tr(\beta_j(\alpha_{i'} - \alpha_i))}{\alpha_{i'} - \alpha_i}\right) f_{i'}.
\end{align*}
Note that the right-hand side of the last equality does not depend upon the lost information at node $i$.

Recall that the trace-dual basis gives an equation for $f_i$.
\begin{equation}\label{full repair}f_i = \sum_{j = 1}^l  \Tr(\beta_j f_i) \widehat{\beta_j}.\end{equation}
The value $\Tr(\beta_j f_i)$ may be rewritten as
\begin{align*}
\Tr(\beta_j f_i) &= \Tr\left(-\sum \limits_{i' \neq i} \left(\frac{\Tr\left(\beta_j(\alpha_{i'} - \alpha_i)\right)}{\alpha_{i'} - \alpha_i}\right) f_{i'}\right) \\
&= -\sum_{i' \neq i}\Tr\left(\frac{f_{i'}}{\alpha_{i'} - \alpha_i} \Tr\left(\beta_j(\alpha_{i'} - \alpha_i)\right)\right)\\
&= -\sum_{i' \neq i} \Tr\left( \frac{f_{i'}}{\alpha_{i'} - \alpha_i} \right) \Tr\left(\beta_j(a_{i'} - \alpha_i)\right).
\end{align*}
The last equality follows since the trace $\Tr: S \to R$ is $R$-linear and surjective. Thus, we obtain a linear exact repair scheme since we may write $f_i = \sum \limits_{j = 1}^l\Tr(\beta_jf_i) \widehat \beta_j$.

Note that the elements $\Tr\left(\beta_j(a_{i'} - \alpha_i)\right)$ do not depend upon the encoded message $f(\mathcal{T})$, so the linear exact repair scheme needs only to send $\Tr\left(\frac{f_{i'}}{\alpha_{i'} - \alpha_i}\right)$, which are elements of $R$. This amounts to a total communication of $n-1$ symbols of $R$, the claimed repair bandwidth.
\end{proof}
To illustrate the above theorem, consider the following example.

\begin{ex}
Let $R = \Z_4$ and $S = \gr{2}{2}{3}$, where $\gr 223 = \Z_4[x]/\langle f(x)\rangle$ with $f(x) = x^3 + 3x + 3$. The ring $S$ has a residue field isomorphic to $\F_2[x]/\langle x^3 + x + 1\rangle$. 
Then, $\mathcal T = \{0\}\cup\{\gamma^i : 0 \leq i \leq 6\}$, where $\gamma$ satisfies $\gamma^3 = \gamma + 1$. The trace is given by the polynomial $\Tr^S_R(x) = x + x^2 + x^4$. Certainly $\B = \{1, \gamma, \gamma^2\}$ is a basis for $S$ over $R$ which has a corresponding dual basis $\widehat{\B} = \{1, \gamma^2, \gamma\}$. Note that while $\widehat \B$ is a re-ordering of $\B$, the ordering matters for a trace-dual basis, i.e. $\B$ and $\widehat{\B}$ are distinct.

Now define a full length Reed-Solomon code $\C = \RS(8,2)$, over $S$, which has evaluation set $\balpha = (0, 1, \gamma, \ldots, \gamma^6)$ and rank $2$ over $S$. Note that we take $k = 2$ for simplicity; however, we may choose any $k \leq 7 = 2^{3 \cdot 1}\left(1 - \frac{1}{8}\right)$. Let $c \in \C$ be the codeword associated with the message polynomial $f(x) = \gamma^2 + x$, that is,
\[c = \left(\gamma^2, \gamma^6, \gamma^4, 0, \gamma^5, \gamma, \gamma^3, 1\right).\]
If the first node, $f_0 = \gamma^2$ is lost, then each remaining node calculates and sends $r_i := \Tr_R^S\left(\frac{f_i}{\gamma^{i-1}} \right)$. The vector of these seven values is then 
\[\mathbf{r} = (1, 1, 0 , 0 , 0, 1, 0).\]
Furthermore, for each basis element $\beta_i \in \B$ and each element $\alpha_j \in \alpha$, calculate $\Tr(\beta_i \alpha_j)$
\[\begin{array}{r | c c c c c c c}
 \Tr_R^S(\beta_i \alpha_j) &1 &\gamma &\gamma^2 &\gamma^3 &\gamma^4 &\gamma^5 &\gamma^6\\ \hline
 1 &1 &0 &0 &1 &0 &1 &1\\
 \gamma &0 &0 &1 &0 &1 &1 &1\\
 \gamma^2 &0 &1 &0 &1 &1 &1 &0
\end{array},\] each row we will denote by $\Tr_R^S(\beta_i \balpha^*)$. Hence the theorem allows us to recover
\begin{align*}
f_0 &= \Tr_R^S(1 f_0) 1 + \Tr_R^S(\gamma f_0)\gamma^2 + \Tr_R^S(\gamma^2 f_0) \gamma\\
&= \left(\mathbf{r} \cdot \Tr_R^S(1 \balpha^*)\right) 1 + \left(\mathbf{r} \cdot \Tr_R^S(\gamma \balpha^*)\right) \gamma^2 + \left(\mathbf{r} \cdot \Tr_R^S(\gamma^2 \balpha^*)\right) \gamma\\
&= 0 \cdot 1 + 1 \cdot \gamma^2 + 0 \cdot \gamma\\
&= \gamma^2,
\end{align*}
where the `$\cdot$' in the second line denotes the Euclidean inner product of the two vectors.
\end{ex}
\section{Conclusion and Future Directions}\label{conclusion}
In this article, we introduced the theory of repair schemes for codes over rings. Due to their similarity to finite fields, the primary focus of the manuscript was repair schemes for free maximum distance separable codes over Galois rings and full-length Reed-Solomon codes over Galois rings. Our work on Galois rings extensively utilized the trace function between Galois rings, the existence of a trace-dual basis, and the dual of free Reed-Solomon codes.

As we stated before, the linear exact repair scheme we developed for codes over Galois rings highly depends on the trace function between Galois rings. For more general rings, there is no obvious definition for a trace function, much less a user-friendly concept such as the trace-dual basis. This would be an interesting direction for further work.

\subsection*{Acknowledgement}
We want to thank the reviewers for their careful reading of the manuscript and for providing such insightful comments. Their suggestions have made for a much improved article. 

\end{document}